\documentclass{article}
\usepackage{amsmath}
\usepackage{array}
\usepackage{amsfonts}
\usepackage{amsmath}
\usepackage{rotate,epsfig}
\usepackage{caption}
\usepackage{amsfonts}
\usepackage{lscape}
\usepackage{wasysym}
\usepackage{graphicx}
\usepackage{color}
\usepackage{enumerate}
\usepackage[table]{xcolor}
\definecolor{lightgray}{gray}{0.9}
\usepackage{longtable}
\usepackage[all]{xy}
\usepackage{lscape}
\graphicspath{ {image/} }
\usepackage{amsmath,amsthm,amssymb,graphicx,tikz,fancyhdr,hyperref,natbib}
\usepackage{graphicx,hyperref,fancyhdr,amsmath,amssymb,mathrsfs,color,xspace,pdflscape}
\hypersetup{colorlinks=true,linkcolor=red,citecolor=blue,filecolor=magneta,urlcolor=cyan}
\DeclareMathSizes{1}{5}{5}{5}
\begin{document}

\title{\bf  A ranked-based estimator of the mean past lifetime with its application }
\author{
 \textbf{ Elham Zamanzade, Majid Asadi, Afshin Parvardeh, and Ehsan Zamanzade }\\
{\small Department of Statistics, Faculty of Mathematics and Statistics, }\\ 
{\small University of Isfahan, Isfahan 81746-73441, Iran.}
}
\date{}

\maketitle

\begin{abstract}
{    The mean past lifetime (MPL) is an important tool in reliability and survival analysis for measuring the average time elapsed since  the occurrence of an event,  under the condition  that  the event has occurred before a specific time $t>0$.   This article develops a nonparametric estimator for MPL based on observations collected according to ranked set sampling (RSS) design. It is shown that the  estimator that we have developed is a strongly uniform consistent. It is also proved that the introduced estimator tends to a Gaussian process under some mild conditions. A Monte Carlo simulation study is employed to evaluate the performance of the proposed estimator with its competitor in simple random sampling (SRS).  Our findings show the introduced estimator is more efficient than its counterpart estimator in SRS  as long as the quality of ranking is better than random. Finally, an illustrative example is provided to describe the potential application of the developed estimator in assessing   the average time between the infection and diagnosis in  HIV patients.
  }
\end{abstract}

\newtheorem{theorem}{Theorem} %
\newtheorem{acknowledgement}[theorem]{Acknowledgement} %
\newtheorem{algorithm}[theorem]{Algorithm} \newtheorem{axiom}[theorem]{Axiom}\newtheorem{case}[theorem]{Case} \newtheorem{claim}[theorem]{Claim} %
\newtheorem{conclusion}[theorem]{Conclusion} %
\newtheorem{condition}[theorem]{Condition} %
\newtheorem{conjecture}[theorem]{Conjecture} %
\newtheorem{corollary}[theorem]{Corollary} %
\newtheorem{criterion}[theorem]{Criterion} %
\newtheorem{definition}[theorem]{Definition} %
\newtheorem{example}[theorem]{Example} %
\newtheorem{exercise}[theorem]{Exercise} \newtheorem{lemma}[theorem]{Lemma} %
\newtheorem{notation}[theorem]{Notation} %
\newtheorem{problem}[theorem]{Problem} %
\newtheorem{proposition}[theorem]{Proposition} %
\newtheorem{remark}[theorem]{Remark} \newtheorem{solution}[theorem]{Solution}
\newtheorem{summary}[theorem]{Summary}

\newpage
\section{Introduction}\label{Sec1}
The mean past lifetime (MPL) of a nonnegative random variable $X$  is defined as
\begin{align}  \label{kt}
K(t)\doteq \mathbb{E}\left(t-X|X\leq t \right),
\end{align}
and the functional form of the MPL function $K(t)$ in terms of the CDF $F$ is given by
\begin{align}  \label{Kt}
K(t)&=\dfrac{\int_{0}^t F(x)dx}{F(t)}, \quad t>0,
\end{align}
provided that $F(t)>0$.   \cite{Asadi2}  investigated some of MPL properties in connection with other reliability measures and discussed its estimation based on simple random sampling (SRS).  \cite{Parvardeh} studied the asymptotic behaviour of  the empirical estimator of MPL function based on SRS.

Ranked set sampling (RSS), introduced by \cite{McIntyre}, is a useful sampling technique when exactly quantifying  sample units is relatively difficult (expensive, destructive, or time consuming) but one can use prior information to rank sample units in a small set, without their actual quantifications. The prior information can be taken as judgement ranking mechanism and be result of a visual comparison, an expert opinion or a concomitant variable. To obtain a ranked set sample, one first determines the parameters $k$ and $m$, referred as the set size and cycle size, respectively. One then draws $k$  simple random samples (sets) of size $k$ from population of interest. In the first set, the unit that judged to be smallest  is maintained and measured ($X_{[1]1}$). In the second set, the unit that judged to be second smallest among $k$ units is maintained and measured ($X_{[2]1}$). This process continues until the unit that judged to be largest is maintained and measured ($X_{[k]1}$). These measured values $X_{[r]1},\ r=1,\ldots,k$, construct the first cycle of the sampling. For increasing the size of sampling, this process is replicated $m$ times (cycles) to reach a ranked set sample of size $n=m \times k$, represented by $\lbrace X_{[r]j}, r=1,\ldots,k, j=1,\ldots,m \rbrace $. We call a ranking procedure \textit{consistent} if  
\begin{eqnarray}
F\left(t\right)=\frac{1}{k}\sum_{r=1}^k F_{[r]}\left(t\right),
\label{fundamental}
\end{eqnarray} 
where  $F_{[r]}$ is the CDF a sample unit in a set of size $k$ which  judged to have rank $r$. \cite{Presnell}  showed that the \textit{consistent} assumption in ranking process holds under some mild conditions. Note that if the ranking processes leads to no error, then the distribution of $X_{[r]j}$ is the similar to the distribution of the $r$th order statistic of a sample of size $k$ (for $j=1,\ldots,m$) and the equality (\ref{fundamental}) obviously holds owing to binomial expansion.

Although, at the beginning, \cite{McIntyre} considers the problem of mean estimation in  agricultural context, through the years, many standard statistical topics have been discussed in RSS and substantial works have been done in different fields. For example, the problem of estimation of the population mean is considered by \cite{Takahasi}. \cite{Stokes(1988)},  and \cite{Huang}  discussed cumulative distribution function (CDF) estimation in RSS. The problem of the variance and the proportion estimation in RSS are addressed by \cite{Stokes(1980)}, \cite{MacEachern(2002)} and  \cite{Chen (2007)}, \cite{Zamanzade&Mahdizadeh(2017)}, respectively. \cite{mahdi1} and \cite{mahdi2} discussed reliability estimation in RSS and \cite{Zamanzade(2019)} considered estimation of the mean residual life. \cite{Al-Omari}, \cite{Haq1}, \cite{Haq2}, \cite{Haq3} studied constructing statistical control charts using RSS. \cite{Samawi1}, and \cite{Samawi2} applied RSS on logisitic regression analysis. \cite{Chen2019}, \cite{Qian}, \cite{He} and \cite{He2021} studied parametric estimation in RSS, and \cite{wang1}, \cite{wang2} described how RSS can be used in multi-stage sampling designs. We refer the interested reader to recent survey paper of \cite{Wolfe (2012)}.

 In Section \ref{Sec2}, we propose an RSS-based estimator for MPL function  and investigate some of its asymptotic properties.  In Section \ref{Sec 3}, we compare the the estimator in RSS with its competitor in SRS. The comparison results show the preference of the introduced methodology. In Section \ref{Sec 4}, a potential application of the proposed method in practice is illustrated using a medical example. Section \ref{Sec 5} is devoted to some concluding remarks.

\section{The proposed estimator} \label{Sec2}
Let $X_1,\ldots, X_n$ be a simple random sample of size $n$ with  CDF function $F$. The empirical estimator of $K(t)$ based on a simple random sample of size $n$ is given by

\begin{align}  \label{Krss}
K_{SRS}(t)&=\dfrac{\sum_{r=1}^n(t-X_r)\mathbb{I}(X_r\leq t)}{\sum_{r=1}^n(X_r\leq t)%
},
\end{align}
where $\mathbb{I}\left(.\right)$ is the usual indicator function. 

The asymptotic behaviour of $K_{SRS}(t)$  is studied by \cite{Parvardeh}. He proved that as $n$ goes to infinity,
\begin{equation*}
\left\{ \sqrt{n}\left( K_{SRS}(t)-K(t)\right) ,\text{ }t\geq 0\right\}
\rightarrow Z\text{ \ \ \ }\ {\rm in} \text{ }D([\tau ,T])\text{\ \ \ }\text{for every \ } 0 <\tau<T<\infty,
\end{equation*}

where $D([\tau ,T])$ is the usual $D$ space on $[\tau ,T]$ with the
Skorokhod topology (\citealp{Billingsley}) and $Z=\{Z(t),$ $t\geq 0\}$ is a
mean zero Gaussian process with variance function $\sigma_{SRS}^{2}(t)=\sigma ^{2}(t)/F(t)$ where $\sigma ^{2}(t)=\mathbb{V}\left(t-X|X<t\right)=\dfrac{2\int_{0}^t(t-x)F(x)dx}{F(t)}-K^2(t)$.\\

 Let $\lbrace X_{[r]j}; r=1,\ldots,k; j=1, \ldots, m \rbrace$ be a  ranked set sample of size $n=mk$ from the population of interest. To estimate the parameter $K(t)$, one  can replace CDF $F$ in equation (\ref{Kt}) with its empirical counterpart,   $F_{RSS}\left(t\right)=\frac{1}{mk} \sum_{r=1}^k \sum_{j=1}^m \mathbb{I}\left(X_{[r]j}\leq t \right)$,   in RSS. This  leads to the following estimator
\begin{align}  \label{Krss}
K_{RSS}(t)&=W\sum_{r=1}^k U_r,
\end{align}
in which $W=\dfrac{\mathbb{I}(V_1+ \ldots +V_k>0)}{V_1+\ldots +V_k}$, $V_r=\sum_{j=1}^{m}\mathbb{I}(X_{[r]j}\leq t),$ and $U_r=\sum_{j=1}^{m}(t-X_{[r]j})\mathbb{I}(X_{[r]j}\leq t)$ for $r=1, \ldots, k$. The entire expression in
equation (\ref{Krss}) is taken to be $0$ if $\mathbb{I}(V_1+ \ldots +V_k>0)=0$, i.e.,
if $V_1+\ldots +V_k=0.$ Note that conditional on $\mathbf{V}=(V_1, \ldots,
V_k)$, the sums $U_r$, $r=1,\ldots,k$%
, are independent random variables. Also, for a fixed value of $r$, $%
U_r$ depends on $\mathbf{V}$ only
through $V_r$. Thus, it follows that
\begin{align*}
\mathbb{E}\left[ K_{RSS}(t)|\mathbf{V}\right] =W\sum_{r=1}^{k}\mathbb{E}\left[
U_r\vert V_r\right],
\end{align*}

\begin{align*}
\mathbb{V}\left[ K_{RSS}(t)|\mathbf{V}\right] =W^2\sum_{r=1}^{k}\mathbb{V}\left[
U_r\vert V_r\right].
\end{align*}

Let $K_{[r]}(t)\doteq \mathbb{E}\left( t-X_{[r]}|X_{[r]}\leq t\right)$ and $\sigma^2_{[r]}(t)$ denotes its variance. One can write

\begin{align*}
\left. \mathbb{E}\left[ K_{RSS}(t)\right\vert \mathbf{V}\right] =W\sum_{r=1}^{k}
V_r K_{[r]}(t),
\end{align*}

\begin{align*}
\left. \mathbb{V}\left[ K_{RSS}(t)\right\vert \mathbf{V}\right]
=W^2\sum_{r=1}^{k} V_r \sigma^2_{[r]}(t).
\end{align*}

Since $V_1, \ldots, V_k$ are independent random variables with $V_r\sim
Bin(m, F_{[r]}(t))$, the mean and the variance of $K_{RSS}(t)$ can be
obtained as

\begin{align}
\mathbb{E}\left[K_{RSS}(t)\right]=\left. \mathbb{E}\lbrace \mathbb{E}\left[K_{RSS}(t)\right\vert \mathbf{V}%
\right] \rbrace = \sum_{v_1+\ldots + v_k>0}\dfrac{p( \mathbf{v})}{v_1+\ldots
+ v_k}\sum_{r=1}^k v_r K_{[r]}(t),
\end{align}
where
\begin{equation*}
p(\mathbf{v})=\mathbb{P}\left( \mathbf{V}=\mathbf{v}\right)=\prod_{r=1}^{k}\binom{m}{%
v_{r}} F_{[r]}(t)^{v_{r}}\left[1-F_{[r]}(t) \right]^{m-v_{r}},
\end{equation*}
and

\begin{align*}
\mathbb{V}[K_{RSS}(t)]&=\left. \mathbb{E}\left\lbrace \mathbb{V}\left[K_{RSS}(t)\right\vert%
\mathbf{V}\right] \right\rbrace + \mathbb{V}\left\lbrace \left. \mathbb{E}\left[
K_{RSS}(t)\right\vert \mathbf{V}\right]\right\rbrace \\
&=\sum_{v_1+\ldots + v_k>0}\dfrac{p(\mathbf{v})}{(v_1+\ldots + v_k)^2}%
\sum_{r=1}^k v_r \sigma^2_{[r]}(t) \\
&\qquad+\sum_{v_1+\ldots + v_k>0}\dfrac{p(\mathbf{v})}{(v_1+\ldots + v_k)^2}%
\left\lbrace \sum_{r=1}^k v_r K_{[r]}(t)\right\rbrace ^2 \\
&\qquad -\left\lbrace \sum_{v_1+\ldots + v_k>0}\dfrac{p(\mathbf{v})}{%
(v_1+\ldots + v_k)} \sum_{r=1}^k v_r K_{[r]}(t)\right\rbrace ^2.\\
\end{align*}

The first theorem in this section establishes the consistency of $K_{RSS}(t)$.

\begin{theorem}
Let  $\lbrace X_{[r]j}; r=1,\ldots,k; j=1, \ldots, m \rbrace$ be a ranked set sample of size $n=mk$. Under consistent ranking assumption, if the set size $k$ is fixed and the number of cycles ($m$)
goes to infinity, then
\begin{align*}
\sup _{\tau<t<T}|K_{RSS}(t)-K(t)| \longrightarrow 0, \text{\ \ \ }\text{for every \ } 0 <\tau<T<\infty.
\end{align*}
\end{theorem}

\begin{proof} We have
\begin{eqnarray*}
K_{RSS}-K(t) &=&\frac{1}{k}\sum_{r=1}^k \dfrac{\hat{K}_{[r]}(t)\hat{F}_{[r]}(t)}{F_{RSS}(t)}-\frac{1}{k}\sum_{r=1}^k \dfrac{K_{[r]}(t)F_{[r]}(t)}{F(t)}\\
&=&\frac{1}{k}\sum_{r=1}^k \dfrac{\hat{F}_{[r]}(t)}{F_{RSS}(t)}\left(\hat{K}_{[r]}(t)-K_{[r]}(t)\right)+
\frac{1}{k}\sum_{r=1}^k \dfrac{\hat{F}_{[r]}(t)}{F_{RSS}(t)}K_{[r]}(t)-\frac{1}{k}\sum_{r=1}^k \dfrac{F_{[r]}(t)}{F_{RSS}(t)}K_{[r]}(t)\\
&&+\frac{1}{k}\sum_{r=1}^k \dfrac{F_{[r]}(t)}{F_{RSS}(t)}K_{[r]}(t)-\frac{1}{k}\sum_{r=1}^k\dfrac{F_{[r]}(t)}{F(t)}K_{[r]}(t)\\
&=&\frac{1}{k}\sum_{r=1}^k \dfrac{\hat{F}_{[r]}(t)}{F_{RSS}(t)}\left(\hat{K}_{[r]}(t)-K_{[r]}(t)\right)+
\dfrac{1}{kF_{RSS}(t)}\sum_{r=1}^k K_{[r]}(t)\left(\hat{F}_{[r]}(t)-F_{[r]}(t)\right)\\
&&+\left( \dfrac{1}{F_{RSS}(t)}-\dfrac{1}{F(t)}    \right) \frac{1}{k}\sum_{r=1}^k K_{[r]}(t)F_{[r]}(t),\\
\end{eqnarray*}
where $\hat{K}_{[r]}(t)$ and $\hat{F}_{[r]}(t)$ are empirical estimates of $K_{[r]}(t)=\mathbb{E}\left(t-X_{[r]}|X_{[r]}\leq t\right)$ and $F_{[r]}(t)$, respectively. Thus,  we get
 \begin{eqnarray}
 \label{alms}
 \left\vert K_{RSS}(t)-K(t)\right\vert &\leq & \max_{1\leq r\leq k} \sup _{\tau\leq t \leq T}\left\vert\hat{K}_{[r]}(t)-K_{[r]}(t)\right\vert \nonumber\\
 &&+ \max_{1\leq r\leq k} \sup _{\tau \leq t }\left\vert\hat{F}_{[r]}(t)-F_{[r]}(t)\right\vert \dfrac{1}{kF_{RSS}(t)}\sum_{r=1}^k K_{[r]}(t)  \nonumber\\
&&+\sup _{\tau \leq t } \left\vert \dfrac{1}{F_{RSS}(t)}-\dfrac{1}{F(t)}    \right \vert \frac{1}{k}\sum_{r=1}^k K_{[r]}(t)F_{[r]}(t).
 \end{eqnarray}
Since $\dfrac{1}{kF_{RSS}(t)}\sum_{r=1}^k K_{[r]}(t) $ and $ \frac{1}{k}\sum_{r=1}^k K_{[r]}(t)F_{[r]}(t)$ are both bounded on a finite interval, it follows from \cite{Parvardeh} that  the first and last term of summation (\ref{alms})  tends to zero as $m$ goes to infinity. Moreover,  one can conclude from Glivenko-Cantelli theorem that the second  term of summation (\ref{alms})  tends to zero as $m \rightarrow \infty$. Hence
\begin{equation*}
\lim_{m\rightarrow \infty }\sup_{\tau \leq t\leq T}\left\vert
K_{RSS}(t)-K(t)\right\vert =0,
\end{equation*}
which completes the proof.
\end{proof}

Now, to obtain asymptotic distribution of $K_{RSS}(t)$,  we should first note that 

\begin{equation*}
\left\{ \sqrt{m}\left( F_{RSS}(t)-F(t)\right) ,\text{ }t\geq 0\right\}
\rightarrow \frac{1}{k}\sum_{r=1}^{k}\mathbb{B}_{[r]}\text{\ \ }\ {\rm in}\text{ }D([\tau
,T])\text{\ },
\end{equation*}
for every $0<\tau<T<\infty$, where $\mathbb{B}_{[r]}$ is a  Gaussian process with mean zero and covariance function (\citealp{Shorack})
\begin{equation*}
\mathbb{E}(\mathbb{B}_{[r]}(t_{1})\mathbb{B}_{[r]}(t_{2}))=F_{[r]}(t_{1}\wedge
t_{2})-F_{[r]}(t_{1})F_{[r]}(t_{2}),\ \ \ t_{1},t_{2}\geq 0.
\end{equation*}

The following theorem provides the asymptotic distribution of $K_{RSS}$.

\begin{theorem}
Let $\left\lbrace X_{[r]j}; r=1,\ldots ,k; j=1,\ldots ,m \right\rbrace$ be a ranked set sample of size $n=mk$. Under consistent ranking process assumption, if the set size $k$ is fixed and the number of cycles ($m$) goes to infinity, then
\begin{equation*}
\left\{ \sqrt{n}\left( K_{RSS}-K(t)\right) ,\text{ }t\geq 0\right\}
\rightarrow Z\text{\ }\ {\rm in}\text{ }D([\tau ,T])\text{\ },
\end{equation*}

for every $0<\tau<T<\infty$, where $Z=\{Z(t),$ $t\geq 0\}$ a mean zero Gaussian process of the form
\begin{equation*}
Z^{RSS}(t)=\dfrac{1}{\sqrt{k}F(t)}\left[ \sum_{r=1}^{k}%
\int_{0}^{t}\mathbb{B}_{[r]}(x)dx-K(t)\mathbb{B}_{[r]}(t)\right] \text{ \ \ \ }t\geq 0,
\end{equation*}%
with variance function
\begin{equation*}
\mathbb{V}(Z^{RSS}(t))=\dfrac{1}{kF^{2}(t)}\sum_{r=1}^{k}\left( \sigma _{\lbrack
r]}^{2}(t)F_{[r]}(t)+F_{[r]}(t)\left(1-F_{[r]}(t)\right)\left[K_{[r]}(t)-K(t)\right]^{2}\right), \text{ \ \ \ \ }t\geq 0.
\end{equation*}
\end{theorem}

\begin{proof}
Note that
\begin{eqnarray*}
K_{RSS}(t)-K(t) &=&\dfrac{1}{F_{RSS}(t)}\int_{0}^{t}\frac{1}{k}%
\sum_{r=1}^{k}\hat{F}_{[r]}(x)dx-\dfrac{1}{F(t)}\int_{0}^{t}\frac{1}{k}%
\sum_{r=1}^{k}F_{[r]}(x)dx \\
&=&\frac{1}{k}\dfrac{1}{F_{RSS}(t)}\sum_{r=1}^{k}\int_{0}^{t}\left( \hat{%
F}_{[r]}(x)-F_{[r]}(x)\right) dx \\
&&+\frac{1}{k}\dfrac{1}{F_{RSS}(t)}\sum_{r=1}^{k}%
\int_{0}^{t}F_{[r]}(x)dx-\frac{1}{k}\dfrac{1}{F(t)}\sum_{r=1}^{k}%
\int_{0}^{t}F_{[r]}(x)dx \\
&=&\frac{1}{k}\dfrac{1}{F_{RSS}(t)}\sum_{r=1}^{k}\int_{0}^{t}\left( \hat{%
F}_{[r]}(x)-F_{[r]}(x)\right) dx \\
&&+\frac{1}{k}\left( \dfrac{1}{F_{RSS}(t)}-\dfrac{1}{F(t)}\right)
\sum_{i=1}^{k}\int_{0}^{t}F_{[r]}(x)dx \\
&=&\frac{1}{k}\dfrac{1}{F_{RSS}(t)}\sum_{r=1}^{k}\int_{0}^{t}\left( \hat{%
F}_{[r]}(x)-F_{[r]}(x)\right) dx \\
&&+K(t)F(t)\left( \dfrac{1}{F_{RSS}(t)}-\dfrac{1}{F(t)}\right) .
\end{eqnarray*}%
Therefore, $\left\lbrace Z_{n}^{RSS}(t)=\sqrt{n}\left( K_{RSS}(t)-K(t)\right), t\geq 0 \right\rbrace$ can be written as
\begin{align*}
Z_{n}^{RSS}(t)& =\frac{1}{k}\dfrac{1}{F_{RSS}(t)}\sum_{r=1}^{k} \int_{0}^{t}\sqrt{n}\left(\hat{F}_{[r]}(x)-F_{[r]}(x)\right) dx \\ & +\dfrac{K(t)}{F_{RSS}(t)}\sqrt{n}\left(F(t)-F_{RSS}(t)\right) .
\end{align*}

Hence, the corresponding limiting process $Z^{RSS}$ is obtained as

\begin{align*}
Z^{RSS}(t)& =\frac{1}{\sqrt{k}}\dfrac{1}{F(t)}\sum_{r=1}^{k}%
\int_{0}^{t}\mathbb{B}_{[r]}(x)dx-\dfrac{K(t)}{\sqrt{k}F(t)}\sum_{r=1}^{k}\mathbb{B}_{[r]}(t)
\\
& =\frac{1}{\sqrt{k}}\dfrac{1}{F(t)}\left[ \sum_{r=1}^{k}%
\int_{0}^{t}\mathbb{B}_{[r]}(x)dx-K(t)\mathbb{B}_{[r]}(t)\right] ,\text{ \ \ \ }t\geq 0.
\end{align*}

The variance function of $Z^{RSS}$ is given by

\begin{eqnarray*}
\mathbb{V} (Z^{RSS}(t))&=& \frac{1}{k}\dfrac{1}{F^2(t)}Cov\left(\sum_{r=1}^k
\int_{0}^t \mathbb{B}_{[r]}(x)dx-K(t)\mathbb{B}_{[r]}(t),\sum_{r=1}^k \int_{0}^t
B_{[r]}(y)dy-K(t)B_{[r]}(t)\right) \\
&=&\sum_{r=1}^k \int_{0}^t \int_{0}^t F_{[r]}(x\wedge
y)-F_{[r]}(x)F_{[r]}(y)dxdy-2K(t)\sum_{r=1}^k\int_{0}^t F_{[r]}(x\wedge
t)-F_{[r]}(x)F_{[r]}(t)dx \\
&&+K^2(t)\sum_{r=1}^k F_{[r]}(t)(1-F_{[r]}(t)) \\
&=&\sum_{r=1}^k \int_{0}^t \int_{0}^t F_{[r]}(x\wedge
y)-F_{[r]}(x)F_{[r]}(y)dxdy-2K(t)\sum_{r=1}^k
(1-F_{[r]}(t))K_{[r]}(t)F_{[r]}(t) \\
&&+K^2(t)\sum_{r=1}^k F_{[r]}(t)(1-F_{[r]}(t)) \\
&=&\dfrac{1}{kF^2(t)}\left(\sum_{r=1}^k\left(\left(\sigma^2_{[r]}(t)+K_{[r]}^2(t)%
\right)F_{[r]}(t)-K^2_{[r]}(t)F^2_{[r]}(t)\right)-2K(t)\sum_{r=1}^k
(1-F_{[r]}(t))K_{[r]}(t)F_{[r]}(t)\right. \\
&&\left.+K^2(t)\sum_{r=1}^k F_{[r]}(t)(1-F_{[r]}(t))\right) \\
&=&\dfrac{1}{kF^2(t)}\sum_{r=1}^k \left(
\sigma^2_{[r]}(t)F_{[r]}(t)+F_{[r]}(t)(1-F_{[r]}(t))[K_{[r]}(t)-K(t)]^2 \right).
\end{eqnarray*}
\end{proof}

\section{Comparisons}\label{Sec 3} 
We now compare the performance of the introduced estimator with its competitor in SRS. The comparison is done based on the following imperfect ranking models:
\begin{itemize}
\item \textit{Fraction of random ranking model}: In this model, we assume that with probability $p$ the rank of sample unit with judgement rank $r$ is identified correctly and with probability of $\left(1-p\right)$ is identified randomly. Hence, the CDF of $X_{[r]j}$ in this model is a mixture given by $F_{[r]}=pF_{(r)}+(1-p)F$, where $p\in [0,1]$.

\item \textit{Fraction of neighbor ranking model}: Here we assume that with probability $p$ the rank of sample unit with judgement rank $r$ is identified correctly and with probability of $\frac{\left(1-p\right)}{2}$ is confused with one of its adjacents. Thus, the CDF of $X_{[r]j}$ in this model is a mixture given by $F_{[r]}=\frac{(1-p)}{2}F_{(r-1)}+pF_{(r)}+\frac{(1-p)}{2}F_{(r+1)},$ where $p\in [0,1]$, $%
F_{(0)}=F_{(1)}$ and $F_{(k)}=F_{(k+1)}$.
\end{itemize}

 Note that  the fraction of neighbour ranking model often leads to less severe ranking error than the random ranking model, but it is more conceivable to happen in practice if the ranking is done using personal judgement of an expert.

The comparison is done for three distributions with different MPL curve shapes: Standard exponential distribution (Exp(1)), Weibull distribution with shape parameter 4 and scale parameter 3 (Weibull(4,3)) and Rescaled beta distribution with probability density function $\frac{3}{2}(1-x)^2I_{(0,2)}(x)$ (Rbeta(1,3)). The MPL curves of these three
distributions are depicted in Figure \ref{Ktt}.

\begin{figure}[h!]
\captionsetup{font=scriptsize}  \centering
\includegraphics[scale=0.6]{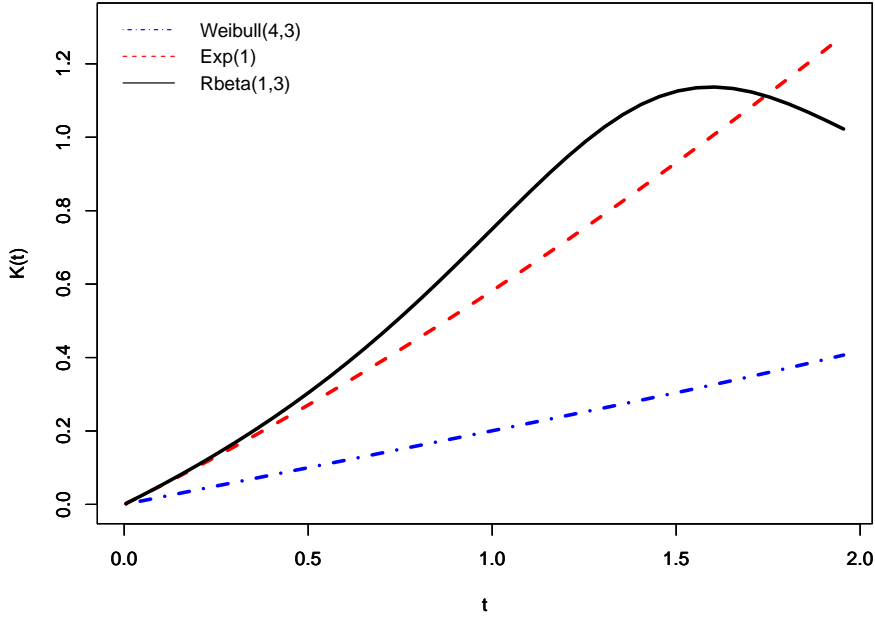}
\caption{ Exact values of $K(t)$  for three parent
distributions Weibull(4,3), Exp(1) and Rbeta(1,3) (represented by dotdash,
dashed and soild, respectively). }
\label{Ktt}
\end{figure}

\subsection{Finite sample size comparisons}
Here, the performance of MPL estimators in RSS and SRS are compared for finite sample sizes using Monte Carlo simulation. To this end, we first set $n \in \lbrace 15, 30, 90 \rbrace$, $k \in \lbrace 3, 5 \rbrace$, and for each combination of $\left(n,k\right)$, we have generated 1,000,000 random samples from both RSS and SRS designs. In order to control ranking quality, the values of $p$ in the imperfect ranking models  are selected from the set $p \in \lbrace 0.2,0.5, 0.8, 1\rbrace$. We have defined relative efficiency (RE) of $K_{RSS}(t)$ to $K_{SRS}(t)$  as the ratio of their mean square errors, i.e., $%
RE\left(t\right)=\frac{MSE%
\left(K_{SRS}(t) \right)}{MSE\left( K_{RSS}(t)\right)}$, and estimated it based on 1,000,000
repetitions for $t \in \lbrace Q_{0.05}, \ldots, Q_{0.95}\rbrace$, where $Q_q
$ is the $q$th quantile of the parent distribution. Note that $RE\left(t\right)>1$ indicates the  preference of $K_{RSS}(t)$ over $K_{SRS}(t)$.

Here, we only report the simulation results for settings with $n=15$ and $k \in \lbrace 3,5 \rbrace$ in Figures \ref{fig:2} and \ref{fig:3} for fraction of random and neighbor ranking models, respectively. This is so because our simulation results show that the $RE$ values are not much affected by sample size $n$ when the other parameters are kept fixed (see the simulation results for $n \in \lbrace 15, 30, 90 \rbrace$ and $k \in \lbrace 3,5 \rbrace$ in Figures S1-S3 for fraction of random ranking model and in Figures S4-S6 for fraction of neighbor ranking model in the Supplementary Materials).

\begin{figure}[h!]
\captionsetup{font=scriptsize}  \centering
\includegraphics[scale=0.72]{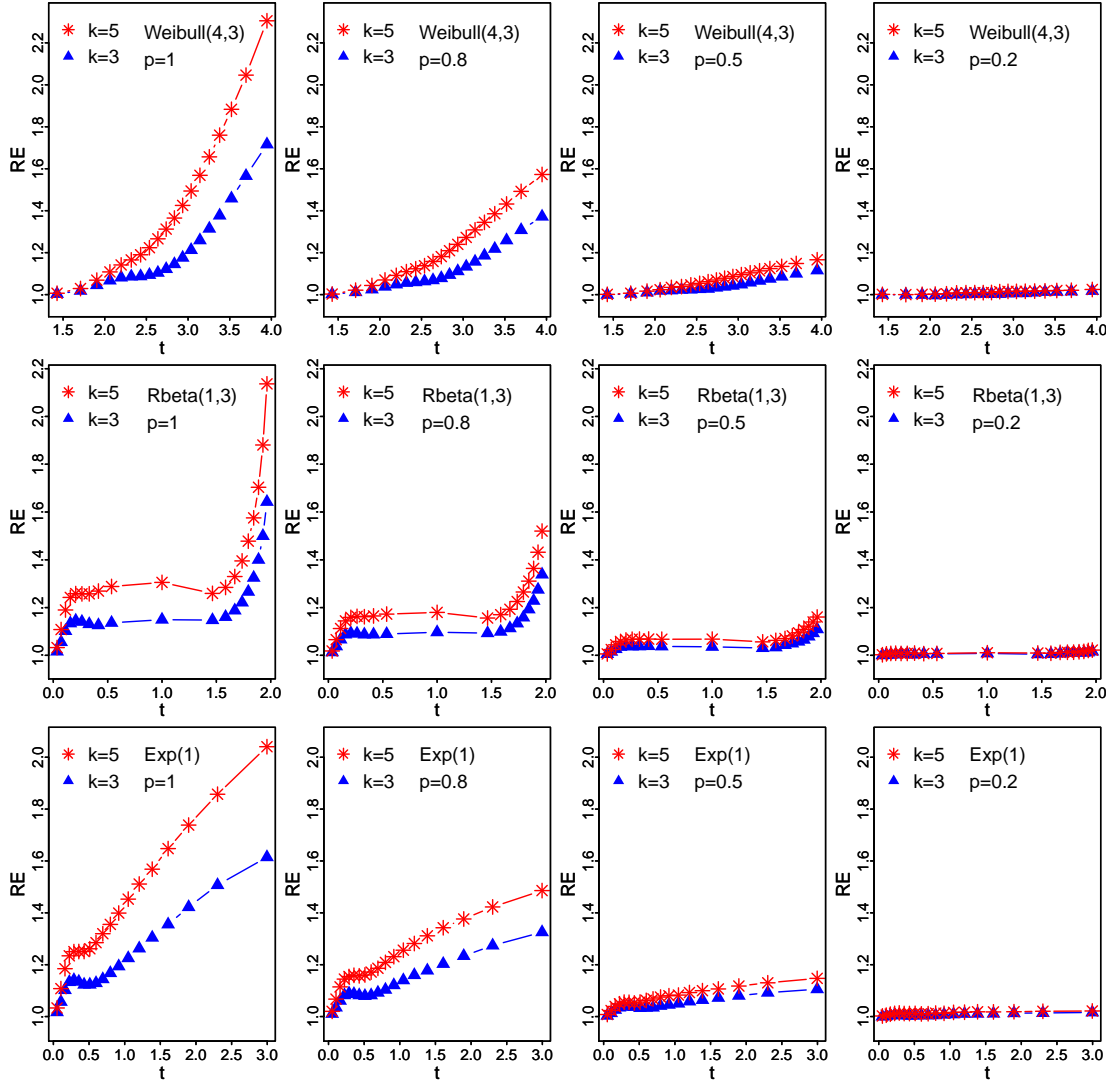}
\caption{ Simulated $RE(t)$  for  $k \in \left\lbrace 3, 5 \right\rbrace$
(represented by $\blacktriangle$, $\ast$ respectively), $n=15$ and $p \in
\left\lbrace 1, 0.8, 0.5, 0.2 \right\rbrace$ under fraction of random
ranking model when population distribution is Weibull(4,3), Rbeta(1,3) and
Exp(1). 
}
\label{fig:2}
\end{figure}

Figure \ref{fig:2} presents the simulation results under fraction of random ranking model. We observe from this figure that although the $RE$ has different patterns for different distributions, it never falls below one which indicates superiority of $K_{RSS}(t)$ to $K_{SRS}(t)$. Furthermore, the relative efficiency increases as the value of set size ($k$) increases while the other parameters are kept fixed, and the efficiency gain is considerable if the quality of ranking is fairly good ($p \geq 0.8$). As one
expects, the relative efficiency decreases as the value of $p$ decreases and $RE$ reaches to almost one when $p=0.2$. This can be justified by the fact that statistical properties of $K_{RSS}(t)$ under fraction of random ranking model with $p=0$ coincides to $K_{SRS}(t)$.

\begin{figure}[h!]
\captionsetup{font=scriptsize}  \centering
\includegraphics[scale=0.72]{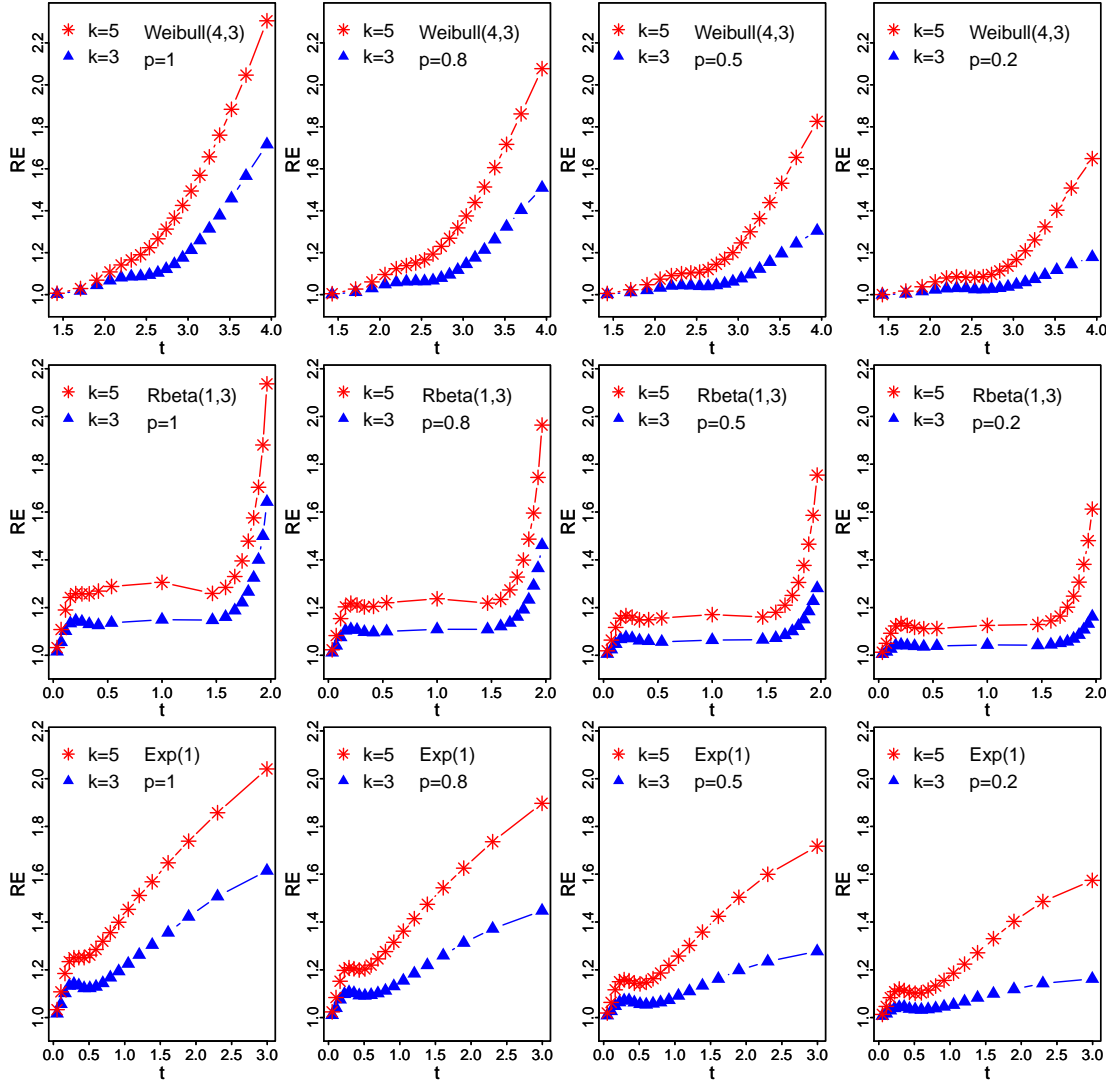}
\caption{ Simulated $RE(t)$   for  $k \in \left\lbrace 3, 5 \right\rbrace$
(represented by $\blacktriangle$, $\ast$ respectively), $n=15$ and $p \in
\left\lbrace 1, 0.8, 0.5, 0.2 \right\rbrace$ under fraction of neighbor
ranking model when population distribution is Weibull(4,3), Rbeta(1,3) and
Exp(1). 
}
\label{fig:3}
\end{figure}

Simulation results under fraction of neighbor ranking model is given in Figure \ref{fig:3}. The pattern of the performance of the estimators in Figure \ref{fig:3} is similar to that of Figure \ref{fig:2}, with an obvious  difference that the $RE$ values in Figure \ref{fig:3} are higher than what we observed in Figure \ref{fig:2}, specially for small values of $p$. This can be justified by the fact that fraction of neighbor ranking model leads to less serve ranking error than the random ranking model.

\subsection{Asymptotic comparisons}

In the SRS design, \cite{Parvardeh} showed that the $\sqrt{n}(K_{SRS}(t)-K(t))$ goes to a  Gaussian process with mean zero and variance function $\sigma^2_{SRS}(t)=\sigma^{2}(t)/F(t)$ as $n$ goes to infinity where $$\sigma^2(t)=\mathbb{V}(t-X|X<t)=\dfrac{2\int_{0}^t(t-x)F(x)dx}{F(t)}-K^2(t). $$
We define the asymptotic relative efficiency ($ARE$)  as the ratio of  asymptotic variance of $K_{SRS}(t)$ to asymptotic variance of  $K_{RSS}(t)$, i.e.,
\begin{equation*}
ARE\left(t\right)=\frac{\sigma^2_{SRS}(t)}{\sigma^2_{RSS}(t)}.
\end{equation*}

In the next theorem, we show that $ARE\left(t\right)$ never falls below one. 

\begin{theorem}
Let $\lbrace X_{[r]j}; r=1,\ldots,k; j=1, \ldots, m \rbrace$ be ranked set sample of size $n=mk$.	  If the ranking process is consistent,   then $ARE(t) \geq 1.$
\end{theorem}

\begin{proof}
Using the fundamental identity $F\left( t\right) =\sum_{r=1}^{k}F_{[r]}(t)/k$, we have

\begin{align*}
k\sigma^2(t)&=k\left( \dfrac{\int_{0}^t2(t-x)F(x)dx}{F(t)}-K^2(t)\right) \\
&=k\left(\frac{1}{k}\sum_{r=1}^{k}\dfrac{F_{[r]}(t)}{F(t)}\int_{0}^t 2(t-x)%
\dfrac{F_{[r]}(x)}{F_{[r]}(t)}dx -K^2(t)\right) \\
&=\sum_{r=1}^{k}\left( \frac{F_{[r]}(t)}{F(t)}\left(\int_{0}^t2(t-x)\dfrac{%
F_{[r]}(x)}{F_{[r]}(t)}dx-K^2_{[r]}(t)\right) +\dfrac{F_{[r]}(t)}{F(t)}%
K^2_{[r]}(t)\right) - k K^2(t) \\
&=\sum_{r=1}^{k}\left(\frac{K_{[r]}(t)}{F(t)}\sigma_{[r]}^2(t)+\dfrac{K_{[r]}(t)}{%
K(t)}K^2_{[r]}(t)\right) - k K^2(t) \\
&=\sum_{r=1}^{k}\frac{F_{[r]}(t)}{F(t)}\sigma_{[r]}^2(t)+\sum_{r=1}^{k}
\dfrac{K_{[r]}(t)}{K(t)}K^2_{[r]}(t)- K(t)\sum_{r=1}^k\frac{F_{[r]}(t)}{F(t)}%
M_{[r]}(t) \\
&=\sum_{r=1}^{k}\frac{F_{[r]}(t)}{F(t)}\sigma_{[r]}^2(t)+\sum_{r=1}^k\frac{%
F_{[r]}(t)}{F(t)}K_{[r]}(t)\left[K_{[r]}(t)-K(t)\right] \\
&=\sum_{r=1}^{k}\frac{F_{[r]}(t)}{F(t)}\sigma_{[r]}^2(t)+\sum_{r=1}^k\frac{%
F_{[r]}(t)}{F(t)}\left[K_{[r]}(t)-K(t)\right]^2. \\
\end{align*}
Therefore, it follows that\bigskip
\begin{align*}
k\sigma^2_{SRS}(t)F^2(t)&=k\sigma^2(t)F(t) \\
&=\sum_{r=1}^{k}F_{[r]}(t)\sigma_{[r]}^2(t)+%
\sum_{r=1}^kF_{[r]}(t)[K_{[r]}(t)-K(t)]^2 \\
&\geq\sum_{r=1}^{k}F_{[r]}(t)\sigma_{[r]}^2(t)+%
\sum_{r=1}^kF_{[r]}(t)(1-F_{[r]}(t))[K_{[r]}(t)-K(t)]^2 \\
&=k\mathbb{V}(F(t)Z_{RSS}(t)) \\
&=k\sigma^2_{RSS}(t)F^2(t),
\end{align*}
and this completes the proof.
\end{proof}

To see the amount of asymptotic efficiency gain obtained using $K_{RSS}$ instead of $K_{SRS}$, we have computed the $ARE\left(t\right)$ for $k \in \lbrace 3, 5\rbrace$ and three different distributions, under two above imperfect ranking models with $p \in \lbrace 0.2,0.5,0.8,1 \rbrace$. Here, for brevity,  we only report results for the perfect ranking case ($p=1$) in Figure \ref{fig:4} and refer the interested reader to see complete results in Figures S7-S8 in the Supplementary Material for fraction of random ranking and neighbour ranking models, respectively.

\begin{figure}[h!]
\captionsetup{font=scriptsize}  \centering
\includegraphics[width=15cm,height=11cm]{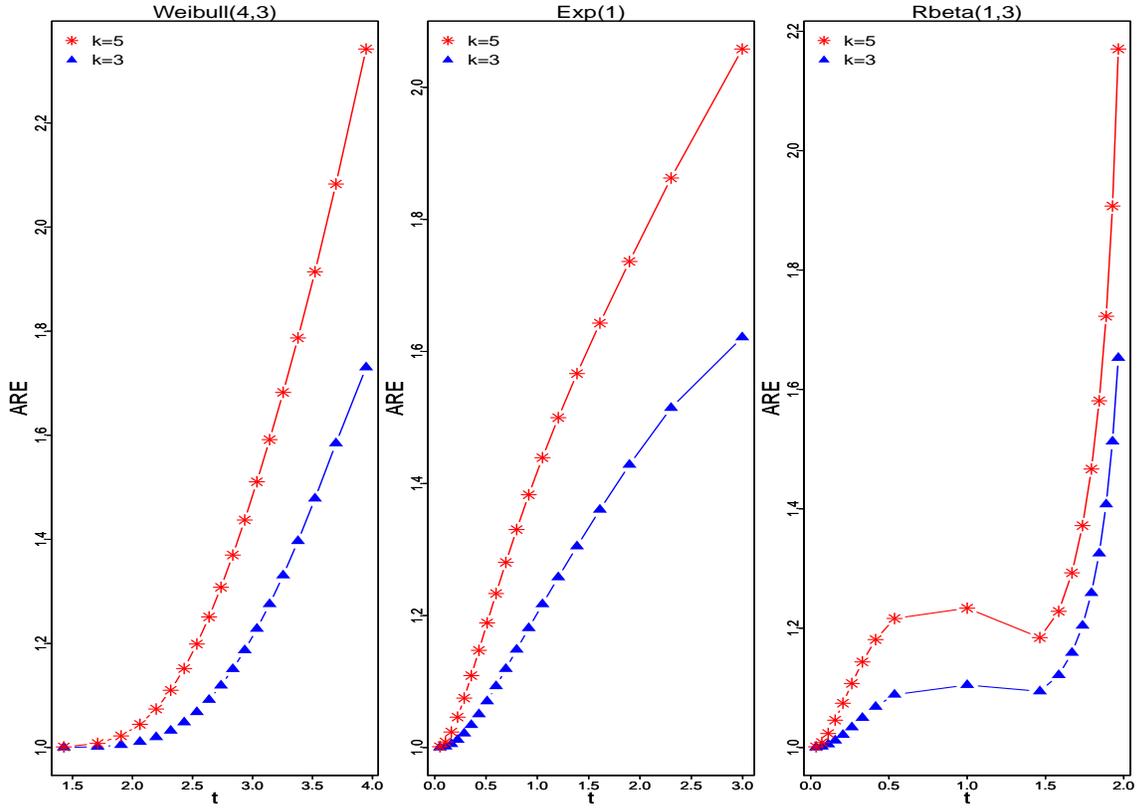}
\caption{ Exact values of ARE(t)  for $k \in \left\lbrace 3, 5 \right\rbrace$ (represented
by $\blacktriangle$, $\ast$, respectively) when the parent distributions are
Weibull(4,3), Exp(1) and Rbeta(1,3) under perfect ranking. }
\label{fig:4}
\end{figure}

Figure \ref{fig:4} presents the $ARE\left(t\right)$ for three different distributions under assumption of perfect ranking. We observe from this figure that the pattern of the performance of the estimator is very close to what we observed for finite sample sizes when the ranking is perfect ($p=1$). This is also true for imperfect ranking case (see Figures S7-S8 in the Supplementary Material). Thus, the $RE\left(t\right)$ values can be well approximated by $ARE\left(t\right)$ ones.

\section{Estimation of the gap between HIV transmission and diagnosis }\label{Sec 4}   
The human immunodeficiency virus (HIV) is a dangerous virus that weakens the body natural defence system against illness (immune system) by destroying a type of white blood cells in the body called CD4 cells. As more CD4 cells are destroyed by the HIV virus, the patient's immune system becomes weaker so he/she finds it harder to fight off infections. The late-stage of HIV is called acquired immune deficiency syndrome (AIDS) in which the body's immune system goes to severely weaker condition such that it cannot defend itself at all. Usually HIV leads to AIDS in 15 years if it does not treated appropriately.  According to the World Health Organization (WHO), more than $95\%$ of HIV infections happen in developing countries. This disease implies a catastrophe not only for the individuals and households affected, but also for the entire nation, as it is likely to lead to an intensification of poverty, push some non-poor into poverty and some of the very poor into destitution. 

 Many people in developing countries who suffer from HIV are not aware that they are infected by the virus. This is so because valid  tests for HIV detection may not be available for public in these countries, or they could be very expensive.   On the other hand, the unaware people of their HIV contribute much larger in HIV transmission. For example, \cite{Hall2012} estimated that around half of HIV transmissions are due to people who are unaware of their HIV infections. Therefore, it is crucial for both government and health organizations  to use an efficient method for estimating the gap between HIV infection and diagnosis.  Suppose that at time $t$, a person gets a medical test to check out about HIV and the test result is positive. Let the random variable $X$ be the infection time, then we know that $X \leq t$. Thus  $K\left(t\right)=\mathbb{E}(t-X|X\leq t)$ is the parameter of interest.

 Note that although obtaining time of HIV infection is difficult because most patients have an established infection of unknown duration at diagnosis, a medical expert can simply rank the sample units according to the infection time using his personal judgement, interviewing with test subjects or their health status. Thus, RSS scheme  can be regarded as an alternative for  SRS for estimating MPL of people living with HIV at the diagnosis time $t$. Assume that one is interested in estimating MPL function of people living with HIV at the diagnosis time $t$ based on an RSS sample. However, since an RSS sample of infection time of patients living with HIV is not available, we simulate an RSS sample with set size $5$ and cycle size $6$ from Gamma distribution with scale parameter 40 and shape parameter 8 under fraction of random ranking model with $p=0.8$ and round it to its nearest integer to represent the infection time of HIV in weeks. The data set is presented in Table \ref{tab1} and the estimated MPL function along with its $95\%$ normal approximation (NA) confidence interval (CI) is shown in Figure \ref{fig:5}, where the variance of $K_{RSS}(t)$ is estimated by replacing each of its components with its empirical counterpart.

\begin{table}[]
\caption{A simulated RSS data set for the infection time of HIV in weeks. }
\label{tab1}\centering
\begin{tabular}{|clclclclc|}
\hline
& Rank 1 & Rank 2 & Rank 3 & Rank 4 & Rank 5 &  &  &  \\ \hline
Cycle 1 & 192 & 277 & 293 & 329 & 376 &  &  &  \\
Cycle 2 & 158 & 256 & 236 & 411 & 508 &  &  &  \\
Cycle 3 & 137 & 280 & 451 & 2447 & 478 &  &  &  \\
Cycle 4 & 171 & 229 & 143 & 268 & 248 &  &  &  \\
Cycle 5 & 203 & 232 & 238 & 367 & 363 &  &  &  \\
Cycle 6 & 287 & 274 & 415 & 413 & 323 &  &  &  \\ \hline
\end{tabular}%
\end{table}

\begin{figure}[h!]
\captionsetup{font=scriptsize}  \centering
\includegraphics[width=10cm,height=10cm]{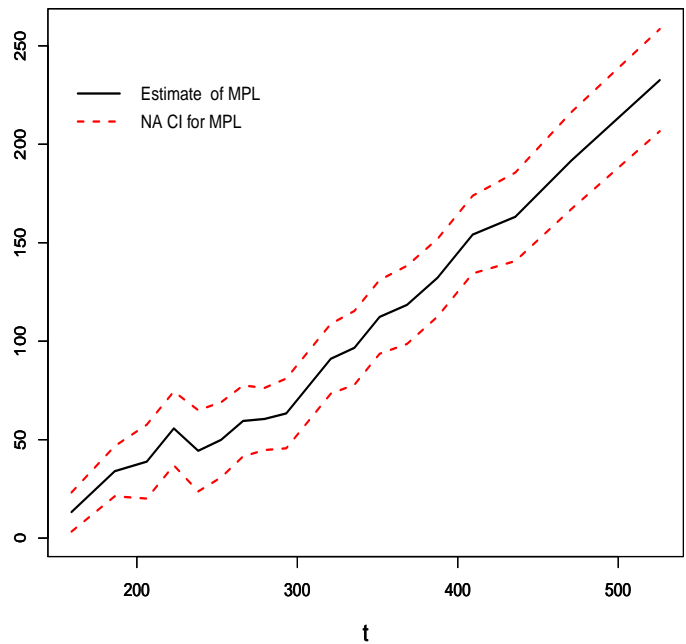}
\caption{ Estimation of MPL function using data set in Table \protect\ref%
{tab1} (represented by solid line) along with its $95\%$ normal
approximation (NA) confidence interval (represented by dash lines). This
figure appears in color in the electronic version of this paper. }
\label{fig:5}
\end{figure}

\section{Conclusion}\label{Sec 5}
The ranked set sampling (RSS) is a sampling plan which is designed to employ auxiliary ranking information for improving the estimation of the population parameters. The auxiliary ranking information can be obtained through eye inspection, subjective judgement, or a cheap concomitant variable. The ranking process is done before any actual measurements on the variable of interest, and leads to select more informative units to include in our sample for measurement. This approach often is used when it is easy and cheap to rank  units in a set without measuring their accurate values.

In this paper, we studied an empirical estimate of mean past lifetime (MPL) based on RSS. We showed that the  estimator is  strongly uniformly consistent estimator and we proved it converges to a Gaussian process under some mild conditions. We then compared the introduced estimator with the empirical estimator in simple random sampling (SRS). To this end, we considered  two imperfect ranking models, which lead to severe and mild ranking errors, four different degrees for ranking quality which move from perfect ranking to almost random ranking, eight  different combination of set size and sample size; and three distributions with different MPL shapes. According to our comparison results, RSS estimator of MPL is more efficient than the SRS one, and the efficiency gain of using  RSS estimator can be sizeable in some certain circumstances. Finally, we illustrated a potential application of the developed procedure in estimating the gap between HIV infection to the diagnosis.


{\small 
\begin{thebibliography}{}
\bibitem[Al-Omari and Haq(2011)]{Al-Omari}Al-Omari, A. I., \& Haq, A., (2011), Improved quality control charts for monitoring the process mean, using double-ranked set sampling methods. \textit{Journal of Applied Statistics}, {\bf 39(4)}, 745-763.
\bibitem[Ahn et al.(2017)]{wang3} Ahn, S.,   Wang, X., \& Lim, J.  (2017). On unbalanced group sizes in cluster randomized designs using balanced ranked set sampling, \textit{Statistics \& Probability Letters}, {\bf 123}, 210-217.

\bibitem[Asadi and Berred(2012)]{Asadi2} Asadi, M., \& Berred, A. (2012). Properties and estimation of the mean past lifetime. \textit{\ Statistics}, \textbf{46}, 405-417.


\bibitem[Billingsley(1999)]{Billingsley} Billingsley, P. (1999). Convergence of probability measures, 2nd edition. Wiley, New York.

\bibitem[Chen et al.(2007)]{Chen (2007)} Chen, H., Stasny, E. A., \& Wolfe, D. A. (2007). Improved procedures for estimation of disease prevalence using ranked set sampling. \textit{Biometrical Journal}, {\bf 49}, 530-538.

\bibitem[Chen et al.(2019)]{Chen2019} Chen, W., Yang, R., Yao, D., Long, C. (2019). Pareto parameters estimation using moving extremes ranked set sampling, To appear in\textit{Statistical Papers}.




\bibitem[Hall et al.(2012)]{Hall2012} Hall, H.I., Holtgrave, D.R., \& Maulsby, C. (2012). HIV transmission rates from persons living with HIV who are aware and unaware of their infection. \textit{AIDS: An official journal of AIDS society}, {\bf 26(7)}, 893-896.

\bibitem[Haq et al.(2013)]{Haq1} Haq, A., Brown, J., Moltchanova, E., \& Al-Omari, A. I. (2013). Partial ranked set sampling design, \textit{Environmetrics}, {\bf 24(3)}, 201-207. 

\bibitem[Haq and Al-Omari(2014)]{Haq2}Haq, A., \& Al-Omari, A. I., (2014), A new Shewhart control chart for monitoring
process mean based on partially ordered judgment subset sampling, \textit{Quality \& Quantity}, {\bf 49(3)}, 1185-1202.

\bibitem[Haq et al.(2014)]{Haq3} Haq, A., Brown, J., Moltchanova, E., \& Al-Omari, A. (2014). Effect of measurement
error on exponentially weighted moving average control charts under ranked set sampling schemes, \textit{Journal of Statistical Computation \& Simulation}, {\bf 85(6)}, 1224-1246. 

\bibitem[He et al.(2020)]{He} He, X., Chen, W., \& Qian, W. (2020). Maximum likelihood estimators of the parameters of the log-logistic distribution, \textit{Statistical Papers}, {\bf 61}, 1875-1892. 


\bibitem[He et al.(2021)]{He2021} He, X., Chen, W., \& Rui, Y. (2021). Modified best linear unbiased estimator of the shape parameter of log-logistic distribution, \textit{Journal of Statistical Computation \& Simulation}, {\bf 91 (2)}, 383-395.

\bibitem[Huang(1997)]{Huang} Huang, J. (1997). Asymptotic properties of the npmle of a distribution function based on ranked set samples. \textit{The Annals of Statistics}, {\bf 25}, 1036-1049.
%

\bibitem[MacEachern(2002)]{MacEachern(2002)} MacEachern, S. N., Ozturk, O., Wolfe, D. A., \& Stark, G. V. (2002). A new ranked set sample estimator of variance. \textit{Journal of the Royal Statistical Society. Series B (Statistical Methodology)}, {\bf 62}, 177-188.

\bibitem[McIntyre(1952)]{McIntyre} McIntyre, G.A. (1952).  A method for unbiased selective sampling using ranked set sampling. \textit{Australian Journal of  Agricultural Research}, {\bf 3}, 385-390.


\bibitem[Mahdizadeh and Zamanzade(2018a)]{mahdi1} Mahdizadeh, M., \&  Zamanzade, E. (2018). A new reliability measure in ranked set sampling, \textit{Statistical Papers}, {\bf 59(3)}, 861-891.

\bibitem[Mahdizadeh and Zamanzade(2018b)]{mahdi2} Mahdizadeh, M., \& Zamanzade, E. (2018). Smooth estimation of a reliability function in ranked set sampling, \textit{Statistics: A Journal of Theoretical \& Applied Statistics}, {\bf 52(4)}, 750-768.

\bibitem[Parvardeh(2015)]{Parvardeh} Parvardeh, A. (2015). A note on the asymptotic distribution of the estimation of the mean past lifetime.\textit{\ Statistical Papers}, \textbf{56(1)}, 205-215.

\bibitem[Presnell and Bohn(1999)]{Presnell} Presnell, B., \& Bohn, L. L. (1999). U-statistics and imperfect ranking in ranked set sampling. \textit{Journal of Nonparametric Statistics}, \textbf{10}, 111-126.


\bibitem[Qian et al.(2021)]{Qian}Qian, W., Chen, W., He, X. (2021). Parameter estimation for the Pareto distribution based on ranked set sampling. \textit{Statistical Papers}, {\bf 62}, 395-417.

\bibitem[Samawi et al.(2017)]{Samawi1}Samawi, H. M., Rochani, H., Linder, D., \& Chatterjee, A. (2017). More efficient logistic analysis using moving extreme ranked set sampling. \textit{Journal of Applied Statistics}, {\bf 44(4)}, 753-76.

\bibitem[Samawi et al.(2018)]{Samawi2}Samawi, H. M., Helu, A., Rochani, H., Yin, J., Yu, L., \& Vogel, R. (2018). Reducing sample size needed for accelerated failure time model using more efficient sampling methods. \textit{Journal of Statistical Theory \& Practice}, {\bf 12(3)}, 530-541.

\bibitem[Shorack and Wellner(2009)]{Shorack}Shorack, G. R., \& Wellner, J. A. (2009). \textit{Empirical processes with applications to statistics}. 59, Siam, Philadelphia.

\bibitem[Stokes(1980)]{Stokes(1980)} Stokes, S. L. (1980). Estimation of variance using judgement ordered ranked set samples.\textit{ Biometrics}, {\bf 36}, 35-42.

\bibitem[Stokes and Sager(1988)]{Stokes(1988)} Stokes, S. L.,\&  Sager, T. W. (1988). Characterization of a ranked-set sample with application to estimating distribution functions. \textit{Journal of the American Statistical Association}, {\bf 38}, 374-381.

\bibitem[Takahasi and Wakimoto(1968)]{Takahasi} Takahasi,K., \&  Wakimoto, K. (1968). On unbiased estimates of the population mean based on the sample stratified by means of ordering. \textit{Annals of the Institute of Statistical Mathematics}, {\bf 20(1)}, 1-31.

\bibitem[Wang et al.(2016)]{wang1} Wang, X., Lim, J.,\& Stokes, S. L. (2016). Using ranked set sampling with cluster randomized designs for improved inference on treatment effects. \textit{Journal of the American Statistical Association}, {\bf 111(516)}, 1576-1590.

\bibitem[Wang et al.(2017)]{wang2} Wang, X., Ahn, S., \& Lim, J.  (2017). Unbalanced ranked set sampling in cluster randomized studies. \textit{Journal of Statistical Planning \& Inference}, {\bf 187}, 1-16.

\bibitem[Wolfe(2012)]{Wolfe (2012)} Wolfe, D. A. (2012). Ranked set sampling: its relevance and impact on statistical inference. \textit{ISRN Probability \& Statistics}. Article ID 568385, 1-32.

\bibitem[Zamanzade and Mahdizadeh(2017)]{Zamanzade&Mahdizadeh(2017)} Zamanzade, E., \& Mahdizadeh, M. (2017). A more efficient proportion estimator in ranked set sampling. \textit{Statistics \& Probability Letters}, \textbf{129}, 28-33.

\bibitem[Zamanzade et al.(2019)]{Zamanzade(2019)} Zamanzade, E.,  Parvardeh, A., \& Asadi., M. (2019). \textit{Estimation of mean residual life based on ranked set sampling}. \textit{Computational Statistics \& Data Analysis}, {\bf 135}, 35-55.

%

\end{thebibliography}
}
\end{document}